\documentclass[journal]{IEEEtran}

\usepackage{amsmath}
\usepackage{amsthm,amsfonts,paralist,mathtools}

\usepackage{tikz}
\usetikzlibrary{matrix}

\newcommand{\ket}[1]{|{#1}\rangle}
\newcommand{\bra}[1]{\langle{#1}|}

\DeclareMathOperator{\Tr}{tr}
\DeclareMathOperator{\trace}{Tr}
\DeclareMathOperator{\wt}{wt}

\newtheorem{proposition}{Proposition}
\newtheorem{theorem}[proposition]{Theorem}
\newtheorem{lemma}[proposition]{Lemma}

\theoremstyle{definition}
\newtheorem{example}[proposition]{Example}

\begin{document}
\nocite{*}
%
\title{Infinite Families of Quantum-Classical\\ Hybrid Codes}
%
%
%

\author{Andrew~Nemec
        ~and~Andreas~Klappenecker
\thanks{This paper was presented in part at the 2018 International Symposium on Information Theory \cite{Nemec2018}.}%
\thanks{This research was supported in part by a Texas A\&M University T3 grant.}%
\thanks{A. Nemec and A. Klappenecker are with the Department
of Computer Science \& Engineering, Texas A\&M University, College Station,
TX, 77843 USA e-mail: (nemeca@tamu.edu; klappi@cse.tamu.edu).}}

\maketitle

\begin{abstract} 
Hybrid codes simultaneously encode both quantum and classical information into physical qubits. We give several general results about hybrid codes, most notably that the quantum codes comprising a genuine hybrid code must be impure and that hybrid codes can always detect more errors than comparable quantum codes. We also introduce the weight enumerators for general hybrid codes, which we then use to derive linear programming bounds. Finally, inspired by the construction of some families of nonadditive codes, we construct several infinite families of genuine hybrid codes with minimum distance two and three.  
\end{abstract}


%
\IEEEpeerreviewmaketitle

\section{Introduction}

Hybrid codes simultaneously encode classical and quantum information into quantum digits such that the information is protected against errors when transmitted through a quantum channel. The simultaneous transmission of classical and quantum information was first investigated by Devetak and Shor \cite{Devetak2005}, who characterized the set of admissible rate pairs. Notably, they showed that, at least for certain small error rates, time-sharing a quantum channel is inferior to simultaneous transmission. Constructions of hybrid codes were first studied by Kremsky, Hsieh, and Brun \cite{Kremsky2008} in the context of entanglement-assisted stabilizer codes and by B{\'e}ny, Kempf, and Kribs \cite{Beny2007a, Beny2007b} who outlined an operator-theoretic construction.

More recently, Grassl, Lu, and Zeng \cite{Grassl2017} gave linear programming bounds for a class of hybrid codes and constructed a number of hybrid stabilizer codes with parameters better than those of hybrid codes constructed from quantum stabilizer codes. In particular, these genuine hybrid codes outperform ``trivial'' hybrid codes regardless of the error rate of the channel. Additional work on hybrid codes has been done from both a coding theory approach \cite{Nemec2018} and from an operator-theoretic approach \cite{Majidy2018}, as well as over a fully correlated quantum channel where the space of errors is spanned by $I^{\otimes n}$, $X^{\otimes n}$, $Y^{\otimes n}$, and $Z^{\otimes n}$ \cite{Li2019}. While they are still relatively unstudied, multiple uses for hybrid codes have already become apparent, including protecting hybrid quantum memory \cite{Kuperberg2003} and constructing hybrid secret sharing schemes \cite{Zhang2011}.

In this paper we give some general results regarding hybrid codes, most notably that at least one of the quantum codes comprising a genuine hybrid code must be impure, as well as show that a hybrid code can always detect more errors than a comparable quantum code. We also generalize the weight enumerators given by Grassl et al. \cite{Grassl2017} for hybrid stabilizer codes to more general nonadditive hybrid codes and use them to derive linear programming bounds. Finally, we give multiple constructions for infinite families of hybrid codes with good parameters. The first of these families are single error-detecting hybrid stabilizer codes with parameters $\left[\!\left[n,n-3\!:\!1,2\right]\!\right]_{2}$ where the length $n$ is odd, where an $\left[\!\left[n,k\!:\!m,d\right]\!\right]_{2}$ hybrid code encodes $k$ logical qubits and $m$ logical bits into $n$ physical qubits with minimum distance $d$. The second is a collection of families of single error-correcting hybrid codes constructed using stabilizer pasting, where we paste together stabilizers from Gottesman's $\left[\!\!\!\:\left[2^{j},2^{j}-j-2,3\right]\!\!\!\:\right]_{\!\!\:2}$ stabilizers codes \cite{Gottesman1996b} and the small distance 3 hybrid codes with $n=7,9,10,11$ from \cite{Grassl2017}. Each of these families of hybrid codes were inspired by families of nonadditive quantum codes, especially those constructed by Rains \cite{Rains1999a} and Yu, Chen, and Oh \cite{Yu2015}.

\section{Hybrid Codes}
A quantum code is a subspace of a Hilbert space that allows for encoded quantum information to be recovered in the presence of errors on the physical qudits. Here our encoded message is a unit vector in the Hilbert space $$H = \bigotimes_{\ell=1}^{n} \mathbb{C}^{q} \cong \mathbb{C}^{q^{n}}.$$ We say a qantum code has parameters $\left(\!\left(n,K\right)\!\right)_{q}$ if and only if it can encode a superposition of $K$ orthogonal quantum states into $n$ quantum digits with $q$ levels.

Now suppose that we want to simultaneously transmit classical and quantum messages. Our goal will be to encode them into the state of $n$ quantum digits that have $q$-levels each, so that the encoded message can be transmitted over a quantum channel. A hybrid code has the parameters $\left(\!\left(n,K\!:\!M\right)\!\right)_{q}$ if and only if it can simultaneously encode one of $M$ different classical messages and a superposition of $K$ orthogonal quantum states into $n$ quantum digits with $q$ levels.

We can understand the hybrid code as a collection of $M$ orthogonal $K$-dimensional quantum codes $\mathcal{C}_{m}$ that are indexed by the classical messages $m\in\left[M\right] \coloneqq \left\{1, 2, \ldots, M\right\}$. If we want to transmit a classical message $m\in\left[M\right]$ and a quantum state $\ket{\varphi}$, then we need to encode $\ket{\varphi}$ into the quantum code $\mathcal{C}_{m}$. We will refer to the each of the quantum codes $\mathcal{C}_{m}$ as \emph{inner codes} and the collection $\mathcal{C}=\left\{\mathcal{C}_{m} \mid m\in\left[M\right]\right\}$ as the \emph{outer code}.

\subsection{Error Detection}\label{seced}

The encoded states will be subject to errors when transmitted through a quantum channel.  Our first task will be to characterize the errors that can be detected by the hybrid code. We will set up a projective measurement that either upon receipt of a state $\ket{\psi}$ in $H$ either (a) returns $\epsilon$ to indicate that an error happened or (b) claims that there is no error and returns a classical message $m$ and a projection of $\ket{\psi}$ onto $\mathcal{C}_{m}$.

Let $P_{m}$ denote the orthogonal projector onto the quantum code $\mathcal{C}_{m}$ for all integers $m$ in the range $1\leq m\leq M$.  For distinct integers $a$ and $b$ in the range $1\leq a,b\leq M$, the quantum codes $\mathcal{C}_{a}$ and $\mathcal{C}_{b}$ are orthogonal, so $P_bP_a=0$. It follows that the orthogonal projector onto $\mathcal{C} =\bigoplus_{m=1}^{M}\mathcal{C}_{m}$ is given by $$P=P_{1}+P_{2}+\cdots+P_{M}.$$ We define the orthogonal projection onto $\mathcal{C}^\perp$ by $P_\epsilon = 1-P$. For the hybrid code $\left\{\mathcal{C}_{m} \mid m\in\left[M\right]\right\}$, we can define a projective measurement $\mathcal{P}$ that corresponds to the set $$\left\{ P_{1}, P_{2}, \dots, P_{M}, P_{\epsilon}\right\}$$ of projection operators that partition unity.

We can now define the concept of a detectable error.  An error $E$ is called \textit{detectable}\/ by the hybrid code $\left\{\mathcal{C}_{m} \mid m\in \left[M\right]\right\}$ if and only if for each index $a, b$ in the range $1\leq a, b \leq M$, we have 
\begin{equation}
\label{hkl}
P_{b} E P_{a} = \begin{cases}
\lambda_{E,a} P_a & \text{if $a=b$}, \\
0 & \text{if $a\neq b$}
\end{cases}
\end{equation}
for some scalar $\lambda_{E,a}$. 

The motivation for calling an error $E$ detectable is the following simple protocol. Suppose that we encode a classical message $m$ and a quantum state into a state $\ket{v_{m}}$ of $\mathcal{C}_{m}$, and transmit it through a quantum channel that imparts the error $E$. If the error is detectable, then measurement of the state $E\ket{v_{m}} = EP_{m}\ket{v_{m}}$ with the projective measurement $\mathcal{P}$ either 
\begin{compactenum}[(E1)]
\item returns $\epsilon$, which signals that an error happened, or 
\item returns $m$ and corrects the error by projecting the state back
  onto a scalar multiple $\lambda_{E,m}\ket{v_{m}} = P_{m}EP_{m}\ket{v_{m}}$ of the state
  $\ket{v_m}$.
\end{compactenum}
The definition of a detectable error ensures that the measurement $\mathcal{P}$ will never return an incorrect classical message $d$, since $P_{d} E P_{m} \ket{v_{m}}= 0$ for all $d\neq m$, so the probability of detecting an incorrect message is zero. An error that is not detectable by the hybrid code can change the encoded classical, the encoded quantum information, or both.

The condition in Equation (\ref{hkl}) is equivalent to the hybrid Knill-Laflamme condition \cite[Theorem 4]{Grassl2017} for detectable errors: an error $E$ is detectable by a hybrid code $\mathcal{C}$ with orthonormal basis states $\left\{\ket{c_{i}^{\left(a\right)}}\mid i\in\left[K\right],a\in\left[M\right]\right\}$ if and only if \begin{equation}\label{klvec}\bra{c_{i}^{\left(b\right)}}E\ket{c_{j}^{\left(a\right)}}=\lambda_{E,a}\delta_{ij}\delta_{ab}.\end{equation} Compared to the original Knill-Laflamme conditions for fully quantum codes \cite{Knill1997} where the scalar only depended on the detectable error, these hybrid conditions allow for scalars $\lambda_{E,a}$ that may depend on both the detectable error $E$ and the classical message $a$, allowing more flexibility in the design of codes. However, this flexibility comes at the price of no longer being able to send a superposition of all of the codewords.

The next proposition shows that hybrid codes can always detect more errors than a comparable quantum code that encodes both classical and quantum information. This is remarkable given that the advantages are much less apparent when one considers minimum distance, see \cite{Grassl2017}. 

\begin{proposition}[{\cite{Nemec2018}}]
\label{detectset}
The subset $\mathcal{D}$ of detectable errors in $B\!\left(H\right)$ of an $\left(\!\left(n,K\!:\!M\right)\!\right)_{q}$  hybrid code form a vector space of dimension $$\dim \mathcal{D} = q^{2n} - (MK)^2 + M.$$ In particular, a $\left(\!\left(n, K\!:\!M\right)\!\right)_{q}$ hybrid code with $M>1$ can detect more errors than an $\left(\!\left(n, KM\right)\!\right)_{q}$ quantum code.
\end{proposition}
\begin{proof}
It is clear that any linear combination of detectable errors is detectable. If we choose a basis adapted to the orthogonal decomposition $H=\mathcal{C} \oplus \mathcal{C}^{\perp}$ with $$\mathcal{C}=\mathcal{C}_{1}\oplus\mathcal{C}_{2}\oplus \cdots \oplus \mathcal{C}_{M},$$ then an error $E$ is represented by a matrix of the form 
$$ 
\left(\begin{array}{cc}
A & R \\ 
S & T
\end{array}\right),
$$
where the blocks $A$ and $T$ correspond to the subspaces $\mathcal{C}$ and $\mathcal{C}^{\perp}$ respectively. Since $E$ is detectable, the $MK\times MK$ matrix $A$ must satisfy $$A = \lambda_{E,1} 1_K \oplus \lambda_{E,2} 1_{K} \oplus \cdots \oplus
\lambda_{E,M} 1_{K},$$ where $1_{K}$ denote a $K\times K$ identity matrix, but $R$, $S$, and $T$ can be arbitrary. Therefore, the dimension of the vector space of detectable errors is given by $q^{2n} - \left(MK\right)^{2} + M$.

In the case of an $\left(\!\left(n, KM\right)\!\right)_{q}$ quantum code, $A$ must satisfy $A=\lambda_{E} 1_{KM}$, so the vector space of detectable errors has dimension $q^{2n}-\left(KM\right)^{2}+1$, which is strictly less than $q^{2n} - \left(MK\right)^{2} + M$ when $M>1$.
\end{proof}

We briefly recall the concept of a nice error basis (see \cite{Klappenecker2002, Klappenecker2003, Knill1996} for further details), so that we can define a suitable notion of weight for the errors. Let ${G}$ be a group of order $q^{2}$ with identity element~1 and $\mathcal{U}\!\left(q\right)$ be the group of $q\times q$ unitary matrices.  A \textit{nice error basis}\/ on $\mathbb{C}^{q}$ is a set ${\cal E}=\{\rho(g)\in {\cal U}(q) \,|\, g\in {G}\}$ of unitary matrices such that
\begin{tabbing}
i)\= (iiiii) \= \kill
\>(i)   \> $\rho(1)$ is the identity matrix,\\[1ex]
\>(ii)  \> $\trace\rho(g)=0$ for all $g\in G\setminus \{1\}$,\\[1ex]
\>(iii) \> $\rho(g)\rho(h)=\omega(g,h)\,\rho(gh)$ for all $g,h\in{G}$,
\end{tabbing}
where $\omega(g,h)$ is a nonzero complex number depending on $(g,h)\in G\times G$; the function $\omega\colon G\times G\rightarrow\mathbb{C}^\times$ is called the factor system of $\rho$. We call $G$ the \textit{index group}\/ of the error basis ${\cal E}$. The nice error basis that we have introduced so far generalizes the Pauli basis to systems with $q\ge 2$ levels. 

We can obtain a nice error basis $\mathcal{E}_n$ on $H\cong \mathbb{C}^{q^n}$ by tensoring $n$ elements of $\mathcal{E}$, so $$ \mathcal{E}_n = \mathcal{E}^{\otimes n} = \{ E_1 \otimes E_2\otimes\cdots \otimes E_n \mid E_k \in \mathcal{E}, 1\le k\le n\}.$$ The weight of an element in $\mathcal{E}_n$ are the number of non-identity tensor components. We write $\wt(E)=d$ to denote that the element $E$ in $\mathcal{E}_n$ has weight $d$. A hybrid code with parameters $\left(\!\left(n,K\!:\!M,d\right)\!\right)_{q}$ has \emph{minimum distance} $d$ if it can detect all errors of weight less than $d$.

\begin{example}
\label{nonaddex}
To construct our nonadditive hybrid code $\mathcal{C}$ we will combine two known degenerate stabilizer codes. The first code $\mathcal{C}_{a}$ is the $\left[\!\left[6,1,3\right]\!\right]_{2}$ code constructed by extending the $\left[\!\left[5,1,3\right]\!\right]_{2}$ Hamming code, see \cite{Calderbank1998}, where the stabilizer is given by $$\left\langle XXZIZI, ZXXZII, IZXXZI, ZIZXXI, IIIIIX\right\rangle.$$ The second code $\mathcal{C}_{b}$ is a $\left[\!\left[6,1,3\right]\!\right]_{2}$ code not equivalent to $\mathcal{C}_{a}$, see \cite{Shaw2008}. Its stabilizer is given by $$\left\langle YIZXXY, ZXIIXZ, IZXXXX, IIIZIZ, ZZZIZI\right\rangle.$$

We can check that the resulting two codes are indeed orthogonal to each other. The resulting code $\mathcal{C}$ is a $\left(\!\left(6,2\!:\!2,1\right)\!\right)_{2}$ nonadditive hybrid code, since there are several errors of weight one such that $P_{b}EP_{a}\neq0$, for example $E=IIIIXI$. This shows that even though $\mathcal{C}_{a}$ and $\mathcal{C}_{b}$ are optimal quantum codes on their own, together they make a hybrid code with an extremely poor minimum distance. Later we will see how to construct hybrid codes with better minimum distances.
\end{example}

\subsection{Genuine Hybrid Codes}

In general, it is not difficult to construct hybrid codes using quantum stabilizer codes. As Grassl et al. \cite{Grassl2017} pointed out, there are three simple constructions of hybrid codes that do not offer any real advantage over quantum error-correcting codes:

\begin{proposition}[{\cite{Grassl2017}}]\label{trivcon}
Hybrid codes can be constructed using the following ``trivial" constructions:
\begin{enumerate}
\item Given an $\left(\!\left(n,KM,d\right)\!\right)_{q}$ quantum code of composite dimension $KM$, there exisits a hybrid code with parameters $\left(\!\left(n,K\!:\!M,d\right)\!\right)_{q}$.
\item Given an $\left[\!\left[n,k\!:\!m,d\right]\!\right]_{q}$ hybrid code with $k>0$, there exists a hybrid code with parameters $\left[\!\left[n,k-1\!:\!m+1,d\right]\!\right]_{q}$.
\item Given an $\left[\!\left[n_{1},k_{1},d\right]\!\right]_{q}$ quantum code and an $\left[n_{2},m_{2},d\right]_{q}$ classical code, there exists a hybrid code with parameters $\left[\!\left[n_{1}+n_{2},k_{1}\!:\!m_{2},d\right]\!\right]_{q}$.
\end{enumerate}
\end{proposition}

We say that a hybrid code is \emph{genuine} if it cannot be constructed using one of the above constructions, following the work of Yu et al. on genuine nonadditive codes \cite{Yu2015}. We also refer to a hybrid stabilizer code that provides an advantage over quantum stabilizer codes as a genuine hybrid stabilizer code. While all known genuine hybrid codes are in fact hybrid stabilizer codes, the linear programming bounds in Section \ref{lpb} do not prohibit genuine nonadditive hybrid codes, and may give us some hints as to their parameters.

Multiple genuine hybrid stabilizer codes with small parameters were constructed by Grassl et al. in \cite{Grassl2017}, all of which have degenerate inner codes. Having degenerate inner codes can allow for a more efficient packing of the inner codes inside the outer code than is possible when using nondegenerate codes, giving a hybrid code with parameters superior to those using the first construction of Proposition \ref{trivcon}. However, they do not exclude the possibility that there is a genuine hybrid code where all of the inner codes are nondegenerate. Here, we show that for a genuine hybrid code, at least one of its inner codes must be impure. Recall that a quantum code is \emph{pure} if trace-orthogonal errors map the code to orthogonal subspaces. A code that is not pure is called \emph{impure}.

\begin{proposition}
Suppose $\mathcal{C}$ is a genuine $\left(\!\left(n,K\!:\!M,d\right)\!\right)_{q}$ hybrid code. Then at least one inner code $\mathcal{C}_{m}$ of the hybrid code $\mathcal{C}$ is impure.
\end{proposition}
\begin{proof}
Seeking a contradiction, suppose that every inner code of the hybrid code $\mathcal{C}$ is pure. For $m\in\left[M\right]$, let $P_{m}$ denote the orthogonal projector onto the $m$-th inner code of the hybrid code $\mathcal{C}$. For every nonscalar error operator $E$ of weight less than $d$, we have 
$$ P_{a} E P_{b} = 0,$$
where $a, b\in\left[M\right]$. Let $P=P_{1} + P_{2} + \cdots + P_{M}$ denote the projector onto the $KM$-dimensional vector space spanned by the inner codes. Then 
$$ PEP=0,$$
so the image of $P$ is an $\left(\!\left(n,KM,d\right)\!\right)_{q}$ quantum code, contradicting that the hybrid code $\mathcal{C}$ is genuine. 
\end{proof}

Since for stabilizer codes the definitions of impure and degenerate codes coincide, genuine hybrid stabilizer codes necessarily require that one of the inner codes is degenerate. Therefore, one of the difficulties in constructing families of genuine codes is finding nontrivial degenerate codes. Unfortunately, there are few known families of impure or degenerate codes, see for example \cite{Aly2006, Aly2007}, and they typically have minimum distances much lower than optimal quantum codes, suggesting they are not particularly suitable to use in constructing genuine hybrid codes.

\subsection{Hybrid Stabilizer Codes}

All of the hybrid codes constructed by Grassl et al. \cite{Grassl2017} were given using the codeword stabilizer (CWS)/union stabilizer framework, see \cite{Cross2009, Grassl2008}, which we will briefly describe here. Starting with a quantum code $\mathcal{C}_{0}$, we choose a set of $M$ coset representatives $t_{i}$ from the normalizer of $\mathcal{C}_{0}$ (we will always take $t_{1}$ to be $I$), and then construct the code $$\mathcal{C}=\bigcup\limits_{i\in\left[M\right]}t_{i}\mathcal{C}_{0}.$$ In the case of hybrid codes, $t_{i}\mathcal{C}_{0}$ are our inner codes and $\mathcal{C}$ is our outer code. If both $\mathcal{C}_{0}$ and $\mathcal{C}$ are stabilizer codes, we say that $\mathcal{C}$ is a hybrid stabilizer code.

The generators that define a hybrid code can be divided into those that generate the quantum stabilizer $\mathcal{S}_{\mathcal{Q}}$ which stabilizes the outer code $\mathcal{C}$ and those that generate the classical stabilizer $\mathcal{S}_{\mathcal{C}}$ which together with $\mathcal{S}_{\mathcal{Q}}$ stabilizes the inner code $\mathcal{C}_{0}$ \cite{Kremsky2008}. The generators that define the $\left[\!\left[7,1\!:\!1,3\right]\!\right]_{2}$ hybrid stabilizer code given in \cite{Grassl2017} are given in (\ref{gen7}), where the generators of $\mathcal{S}_{\mathcal{Q}}$ are given above the dotted line, the generators of $\mathcal{S}_{\mathcal{C}}$ are between the dotted and solid line, the normalizer of the inner code $\mathcal{C}_{0}$ is generated by all elements above the double line, and the normalizer of the outer code is generated by all of the elements. 

\begin{equation}
\label{gen7}
\left(\mkern-5mu
\begin{tikzpicture}[baseline=-.5ex]
\matrix[
  matrix of math nodes,
  column sep=.25ex, row sep=-.25ex
] (m)
{
X & I & I & Z & Y & Y & Z \\
Z & X & I & X & Z & I & X \\
Z & I & X & X & I & Z & X \\
Z & I & Z & Z & X & I & I \\
I & Z & I & Z & I & X & X \\
Z & I & I & I & I & I & X \\
I & I & I & X & Z & Z & X \\
I & I & I & Z & X & X & I \\
I & I & I & I & X & Y & Y \\
};
\draw[line width=1pt, line cap=round, dash pattern=on 0pt off 2\pgflinewidth]
  ([yshift=.2ex] m-5-1.south west) -- ([yshift=.2ex] m-5-7.south east);
\draw[line width=.5pt]
  ([yshift=.2ex] m-6-1.south west) -- ([yshift=.2ex] m-6-7.south east);
\draw[line width=.5pt]
  ([yshift=.22ex] m-8-1.south west) -- ([yshift=.2ex] m-8-7.south east);
\draw[line width=.5pt]
  ( m-8-1.south west) -- ( m-8-7.south east);
\end{tikzpicture}\mkern-5mu
\right)
\end{equation}

Following Kremsky et al. \cite{Kremsky2008}, we will often only include the stabilizer generators, as they are sufficient to fully define the hybrid code, as shown in the following proposition:

\begin{proposition}
\label{hybgenconstr}
Let $\mathcal{C}$ be an $\left[\!\left[n,k\!:\!m,d\right]\!\right]_{p}$ hybrid stabilizer code over a finite field of prime order $p$ with quantum stabilizer $\mathcal{S}_{\mathcal{Q}}$ and classical stabilizer $\mathcal{S}_{\mathcal{C}}=\left\langle g_{1}^{\mathcal{C}}, \dots, g_{m}^{\mathcal{C}}\right\rangle$. Then the stabilizer code $\mathcal{C}_{c}$ associated with classical message $c\in\mathbb{F}_{p}^{m}$ is given by the stabilizer $$\left\langle \mathcal{S}_{\mathcal{Q}}, \omega^{c_{1}}g_{1}^{\mathcal{C}}, \dots, \omega^{c_{m}}g_{m}^{\mathcal{C}}\right\rangle,$$ where $c_{i}$ is the $i$-th entry of $c$ and $\omega$ is a primitive complex $p$-th root of unity.
\end{proposition}
\begin{proof}
There are $p^{k+m}$ codewords stabilized by $\mathcal{S}_{\mathcal{Q}}$. Each of these codewords is an eigenvector of $g_{i}^{\mathcal{C}}$, which naturally partitions the code into $p$ cosets based on eigenvalues. Repeating this with all of the classical generators, we get $p^{m}$ cosets of codewords each of size $p^{k}$. Since $v$ being an eigenvector of $g_{i}^{\mathcal{C}}$ with eigenvalue $\omega^{-1}$ means that it is a $+1$ eigenvector of $\omega g_{i}^{\mathcal{C}}$, therefore each coset is the $+1$ eigenspace of a stabilizer of the form $\left\langle \mathcal{S}_{\mathcal{Q}}, \omega^{c_{1}}g_{1}^{\mathcal{C}}, \dots, \omega^{c_{m}}g_{m}^{\mathcal{C}}\right\rangle$, where the string $c\in\mathbb{F}_{p}^{m}$ can be used to index the stabilizer codes.
\end{proof}

\section{Weight Enumerators and\\ Linear Programming Bounds}

Weight enumerators for quantum codes were introduced by Shor and Laflamme \cite{Shor1997}, and as with their classical counterparts they can be used to give good bounds on code parameters using linear programming, see \cite{Ashikhmin1999, Ketkar2006}. Grassl et al. \cite{Grassl2017} gave weight enumerators and linear programming bounds for hybrid stabilizer codes, but these weight enumerators will not work for nonadditive hybrid codes such as the one given in Example \ref{nonaddex}. In this section, we define weight enumerators for general hybrid codes following the approach of Shor and Laflamme \cite{Shor1997} and Rains \cite{Rains1998} and use them to derive linear programming bounds for general hybrid codes.

\subsection{Weight Enumerators}

For an $\left(\!\left(n,K\!:\!M,d\right)\!\right)_{q}$ hybrid code $\mathcal{C}$ defined by the projector $P=P_{1}+\cdots+P_{M}$ and a nice error base $\mathcal{E}_{n}$ as defined in Section \ref{seced}, we define the two weight enumerators of the code following Shor and Laflamme \cite{Shor1997}: $$A\!\left(z\right)=\sum\limits_{d=0}^{n}A_{d}z^{d}\text{ and }B\!\left(z\right)=\sum\limits_{d=0}^{n}B_{d}z^{d},$$ where the coefficients are given by $$A_{d}=\frac{1}{K^{2}M^{2}}\sum\limits_{\substack{E\in \mathcal{E}_n\\ \wt(E)=d}}\Tr\!\left(EP\right)\Tr\!\left(E^{*}P\right)$$ and $$B_{d}=\frac{1}{KM}\sum\limits_{\substack{E\in \mathcal{E}_n\\ \wt(E)=d}}\Tr\!\left(EPE^{*}P\right).$$

We can also define weight enumerators using the inner code projectors $P_{a}$. Let $$A^{\left(a,b\right)}\!\left(z\right)=\sum\limits_{d=0}^{n}A_{d}^{\left(a,b\right)}z^{d}\text{ and }B^{\left(a,b\right)}\!\left(z\right)=\sum\limits_{d=0}^{n}B_{d}^{\left(a,b\right)}z^{d},$$ where $$A_{d}^{\left(a,b\right)}=\frac{1}{K^{2}}\sum\limits_{\substack{E\in \mathcal{E}_n\\ \wt(E)=d}}\Tr\!\left(EP_{a}\right)\Tr\!\left(E^{*}P_{b}\right)$$ and $$B_{d}^{\left(a,b\right)}=\frac{1}{K}\sum\limits_{\substack{E\in \mathcal{E}_n\\ \wt(E)=d}}\Tr\!\left(EP_{a}E^{*}P_{b}\right).$$ Note that $A^{\left(a,a\right)}\!\left(z\right)$ and $B^{\left(a,a\right)}\!\left(z\right)$ are the weight enumerators of the quantum code associated with projector $P_{a}$. We can then write the weight enumerators for the outer code in terms of the weight enumerators for the inner codes:

\begin{lemma}
The weight enumerators of $\mathcal{C}$ can be written as $$A\!\left(z\right)=\frac{1}{M^{2}}\sum\limits_{a,b=1}^{M}A^{\left(a,b\right)}\!\left(z\right)\text{ and }B\!\left(z\right)=\frac{1}{M}\sum\limits_{a,b=1}^{M}B^{\left(a,b\right)}\!\left(z\right).$$
\end{lemma}
\begin{proof}
By linearity of the projector $P$ we have \begin{align*} A_{d} & = \frac{1}{K^{2}M^{2}}\sum\limits_{\substack{E\in \mathcal{E}_n\\ \wt(E)=d}}\Tr\!\left(EP\right)\Tr\!\left(E^{*}P\right) \\ & = \frac{1}{K^{2}M^{2}}\sum\limits_{\substack{E\in \mathcal{E}_n\\ \wt(E)=d}}\sum\limits_{a,b=1}^{M}\Tr\!\left(EP_{a}\right)\Tr\!\left(E^{*}P_{b}\right) \\ & = \frac{1}{M^{2}}\sum\limits_{a,b=1}^{M}A_{d}^{\left(a,b\right)}. \end{align*} We can then rewrite the weight enumerator as \begin{align*} A\!\left(z\right) & = \sum\limits_{d=0}^{n}A_{d}z^{d} \\ & = \frac{1}{M^{2}}\sum\limits_{d=0}^{n}\sum\limits_{a,b=1}^{M}A_{d}^{\left(a,b\right)}z^{d} \\ & =\frac{1}{M^{2}}\sum\limits_{a,b=1}^{M}A^{\left(a,b\right)}\!\left(z\right). \end{align*} The result for $B\!\left(z\right)$ follows from the same argument.
\end{proof}

While the weight enumerator $B\!\left(z\right)$ is the same as the one introduced by the authors in \cite{Nemec2018}, the weight enumerator $A\!\left(z\right)$ is different. There the $A^{\left(a,b\right)}\!\left(z\right)$ weight enumerators with $a\neq b$ were ignored, causing $A\!\left(z\right)$ and $B\!\left(z\right)$ to not satisfy the MacWilliams identity. The approach presented in this paper is more natural, as it treats both the inner and outer codes as quantum codes. The following result may be found in \cite{Rains1998, Shor1997}, which we include for completeness:

\begin{lemma}[{\cite{Rains1998, Shor1997}}]
\label{cauchyschwarz}
Let $\mathcal{C}$ be a $\left(\!\left(n,K\!:\!M\right)\!\right)_{q}$ hybrid code with weight distributions $A_{d}$ and $B_{d}$. Then for all integers $d$ in the range $0\leq d\leq n$ and all $a\in\left[M\right]$ we have
\begin{enumerate}
\item $0\leq A_{d}\leq B_{d}$
\item $0\leq A_{d}^{\left(a,a\right)}\leq B_{d}^{\left(a,a\right)}$.
\end{enumerate}
\end{lemma}
\begin{proof}
For every orthogonal projector $\Pi:\mathbb{C}^{q^{n}}\rightarrow\mathbb{C}^{q^{n}}$ of rank $K$, we have
\begin{equation*}
0\leq\frac{1}{K^{2}}\sum\limits_{\substack{E\in \mathcal{E}_n\\ \wt(E)=d}}\Tr\!\left(E\Pi\right)\Tr\!\left(E^{*}\Pi\right)
\end{equation*}
by the non-negativity of the trace inner product. Furthermore, we can write this inequality in the form
\begin{align*}
0 & \leq \frac{1}{K^{2}}\sum\limits_{\substack{E\in \mathcal{E}_n\\ \wt(E)=d}}\Tr\!\left(E\Pi\right)\Tr\!\left(E^{*}\Pi\right) \\
& = \frac{1}{K^{2}}\sum\limits_{\substack{E\in \mathcal{E}_n\\ \wt(E)=d}}\left\vert\Tr\!\left(E\Pi\right)\right\vert^{2} \\
& = \frac{1}{K^{2}}\sum\limits_{\substack{E\in \mathcal{E}_n\\ \wt(E)=d}}\left\vert\Tr\!\left(\left(\Pi E\Pi\right)\Pi\right)\right\vert^{2}.
\end{align*}
Using the Cauchy-Schwarz inequality, we obtain
\begin{align*}
0 & \leq \frac{1}{K^{2}}\sum\limits_{\substack{E\in \mathcal{E}_n\\ \wt(E)=d}}\Tr\!\left(\left(\Pi E\Pi\right)\left(\Pi E\Pi\right)^{*}\right)\Tr\!\left(\Pi^{*}\Pi\right) \\
& = \frac{1}{K}\sum\limits_{\substack{E\in \mathcal{E}_n\\ \wt(E)=d}}\Tr\!\left(E\Pi E^{*}\Pi\right).
\end{align*}
Substituting $\Pi=P$ implies (1) and substituting $\Pi=P_{a}$ implies (2).
\end{proof}

The main utility of weight enumerators for quantum codes is that they allow for a complete characterization of the error-correction capability of the code in terms of the minimum distance of the code. In the following proposition, we prove a similar result for the weight enumerators of hybrid codes.

\begin{proposition}
\label{wtenumnecsuf}
Let $\mathcal{C}$ be a $\left(\!\left(n,K\!:\!M\right)\!\right)_{q}$ hybrid code with weight distributions $A_{d}$ and $B_{d}$. Then $\mathcal{C}$ can detect all errors in $\mathcal{E}_{n}$ of weight $d$ if and only if $A_{d}^{\left(a,a\right)}=B_{d}^{\left(a,a\right)}$ for all $a\in\left[M\right]$ and $B_{d}^{\left(a,b\right)}=0$ for all $a,b\in\left[M\right],a\neq b$.
\end{proposition}
\begin{proof}
Recall that an error is detectable by a code if and only if it satisfies the hybrid Knill-Laflamme conditions in Equation (\ref{klvec}), and that a projector onto one of the inner codes $\mathcal{C}_{a}$ may be written as $P_{a}=\sum_{i=1}^{K}\ket{c_{i}^{\left(a\right)}}\bra{c_{i}^{\left(a\right)}}$, where $\left\{\ket{c_{i}^{\left(a\right)}}\mid i\in\left[K\right]\right\}$ is an orthonormal basis for $\mathcal{C}_{a}$. Suppose that all errors of weight $d$ are detectable by $\mathcal{C}$. Then \begin{align*}
A_{d}^{\left(a,a\right)} & = \frac{1}{K^{2}}\sum\limits_{\substack{E\in \mathcal{E}_n\\ \wt(E)=d}} \Tr\!\left(EP_{a}\right)\Tr\!\left(E^{*}P_{a}\right) \\
& = \frac{1}{K^{2}}\sum\limits_{\substack{E\in \mathcal{E}_n\\ \wt(E)=d}}\left\vert\sum_{i=1}^{K}\bra{c_{i}^{\left(a\right)}}E\ket{c_{i}^{\left(a\right)}}\right\vert^{2} \\
& = \sum\limits_{\substack{E\in \mathcal{E}_n\\ \wt(E)=d}}\left\vert\alpha_{E}^{\left(a\right)}\right\vert^{2}.
\end{align*}
Similarly, we have \begin{align*}
B_{d}^{\left(a,a\right)} & = \frac{1}{K}\sum\limits_{\substack{E\in \mathcal{E}_n\\ \wt(E)=d}} \Tr\!\left(EP_{a}E^{*}P_{a}\right) \\
& = \frac{1}{K}\sum\limits_{\substack{E\in \mathcal{E}_n\\ \wt(E)=d}}\sum\limits_{i,j=1}^{K}\left\vert\bra{c_{i}^{\left(a\right)}}E\ket{c_{j}^{\left(a\right)}}\right\vert^{2} \\
& = \frac{1}{K}\sum\limits_{\substack{E\in \mathcal{E}_n\\ \wt(E)=d}}\sum\limits_{i=1}^{K}\left\vert\bra{c_{i}^{\left(a\right)}}E\ket{c_{i}^{\left(a\right)}}\right\vert^{2} \\
& = \sum\limits_{\substack{E\in \mathcal{E}_n\\ \wt(E)=d}}\left\vert\alpha_{E}^{\left(a\right)}\right\vert^{2}.
\end{align*}
Therefore, we have that $A_{d}^{\left(a,a\right)}=B_{d}^{\left(a,a\right)}$. Additionally, if $a\neq b$, then by Equation (\ref{klvec}) we have $\bra{c_{i}^{\left(a\right)}}E\ket{c_{j}^{\left(b\right)}}=0$. Therefore, \begin{align*}
B_{d}^{\left(a,b\right)} & = \frac{1}{K}\sum\limits_{\substack{E\in \mathcal{E}_n\\ \wt(E)=d}}\sum\limits_{i,j=1}^{K}\left\vert\bra{c_{i}^{\left(a\right)}}E\ket{c_{j}^{\left(b\right)}}\right\vert^{2} \\
& = 0.
\end{align*}

Conversely, suppose that (a) $A_{d}^{\left(a,a\right)}=B_{d}^{\left(a,a\right)}$ for all $a\in\left[M\right]$ and (b) $B_{d}^{\left(a,b\right)}=0$ for all $a,b\in\left[M\right],a\neq b$. Condition (a) implies that equality holds for each $E$ in the Cauchy-Schwarz inequality. Therefore, we have that $P_{a}EP_{a}$ and $P_{a}$ must be linearly dependent, so there must be a constant $\alpha_{E}^{\left(a\right)}\in\mathbb{C}$ such that $P_{a}EP_{a}=\alpha_{E}^{\left(a\right)}$, or equivalently, $\bra{c_{i}^{\left(a\right)}}E\ket{c_{j}^{\left(a\right)}}=\alpha_{E}^{\left(a\right)}\delta_{i,j}$, for all errors of weight $d$. Condition (b) implies that $\bra{c_{i}^{\left(a\right)}}E\ket{c_{j}^{\left(b\right)}}=0$ if $a\neq b$, for all errors of weight $d$. Putting these together, we get the hybrid Knill-Laflamme conditions, so all errors of weight $d$ are detectable.
\end{proof}

\subsection{Linear Programming Bounds}\label{lpb}

One of the more useful properties of weight enumerators is that they satisfy the Macwilliams identity \cite{Shor1997}:
\begin{equation}
B^{\left(a,b\right)}\!\left(z\right)=\frac{K}{q^{n}}\left(1+\left(q^{2}-1\right)z\right)^{n}A^{\left(a,b\right)}\!\left(\frac{1-z}{1+\left(q^{2}-1\right)z}\right).
\end{equation}
The MacWilliams identities, along with the results from Lemma \ref{cauchyschwarz} and Proposition \ref{wtenumnecsuf} and the shadow inequalities for qubit codes \cite{Rains1999b} allow us to define linear programming bounds on the parameters of general hybrid codes (see \cite{Ashikhmin1999, Calderbank1998, Rains1998} for linear programming bounds on quantum codes). Let \begin{equation}K_{j}\!\left(r\right)=\sum\limits_{k=0}^{j}\left(-1\right)^{k}\left(q^{2}-1\right)^{j-k}\binom{r}{k}\binom{n-r}{j-k}\end{equation} denote the $q^{2}$-ary Krawtchouk polynomials.

\begin{proposition}
The parameters of an $\left(\!\left(n,K\!:\!M,d\right)\!\right)_{q}$ hybrid code must satisfy the following conditions:
\begin{enumerate}
\item $A_{j}=\frac{1}{M^{2}}\sum\limits_{a,b=1}^{M}A_{j}^{\left(a,b\right)}$
\item $B_{j}=\frac{1}{M}\sum\limits_{a,b=1}^{M}B_{j}^{\left(a,b\right)}$
\item $A_{0}^{\left(a,b\right)}=1$
\item $B_{0}^{\left(a,b\right)}=\begin{cases} 1 & \text{ if } a=b \\ 0 & \text{ if } a\neq b \end{cases}$
\item $A_{j}^{\left(a,a\right)}=B_{j}^{\left(a,a\right)}$, for all $0\leq j<d$
\item $B_{j}^{\left(a,b\right)}=0$, for all $0\leq j<d$, $a\neq b$
\item $0\leq A_{j}^{\left(a,a\right)}\leq B_{j}^{\left(a,a\right)}$, for all $0\leq j\leq n$
\item $0\leq A_{j}\leq B_{j}$, for all $0\leq j\leq n$
\item $0\leq B_{j}^{\left(a,b\right)}$, for all $0\leq j\leq n$
\item $B_{j}^{\left(a,b\right)}=\frac{K}{q^{n}}\sum\limits_{r=0}^{n}K_{j}\!\left(r\right)A_{r}^{\left(a,b\right)}$, for all $0\leq j\leq n$ (MacWilliams Identity)
\item $0\leq\sum\limits_{r=0}^{n}\left(-1\right)^{r}K_{j}\!\left(r\right)A_{r}^{\left(a,b\right)}$, for all $0\leq j\leq n$, for qubit codes (Shadow Inequalities)
\end{enumerate}
\end{proposition}
\begin{proof}
Conditions 1) and 2) follow from the definition of $A_{j}$ and $B_{j}$. The constraints 3) and 4) respectively result from substituting $E=I$ into the definition of $A_{0}^{(a,b)}$ and $B_{0}^{(a,b)}$. 

The Knill-Laflamme error-detecting conditions of the hybrid codes shown in Proposition~\ref{wtenumnecsuf} imply the constraints 5) and 6).

The claims 7) and 8) are a consequence of Lemma~\ref{cauchyschwarz}. Essentially, these two conditions follow from the Cauchy-Schwarz inequalities when applied to the quantum and hybrid projectors, respectively. 

The statement 9) is simply a consequence of the non-negativity of all $B_{j}^{\left(a,b\right)}$. Conditions 10) and 11) follow from the MacWilliams identities \cite{Shor1997} and shadow inequalities \cite{Rains1999b} respectively.
\end{proof}

Note that conditions 10) and 11) imply the MacWilliams identity and shadow inequality respectively for the outer code.

If we consider only hybrid stabilizer codes, we have that all of the weight distributions for the inner quantum codes are identical. This along with our error-detecting condition from Proposition \ref{wtenumnecsuf} give us that a stabilizer hybrid stabilizer code can detect all errors of weight $d$ if and only if $A_{d}^{\left(a,a\right)}=B_{d}^{\left(a,a\right)}=B_{d}$. Additionally, straightforward calculations give us the missing piece of the nested code condition, $A_{d}\leq A_{d}^{\left(a,a\right)}$ for all $d$. Thus we recover the linear programming bounds of Grassl et al. when we restrict our bounds to hybrid stabilizer codes, with the exception that we have the additional constraint of the shadow inequality for the outer code. This constraint strengthens the bounds found in Table I of \cite{Grassl2017} and rules out the possibility, for example, of $\left[\!\left[10,4\!:\!1,3\right]\!\right]_{2}$, $\left[\!\left[12,5\!:\!1,3\right]\!\right]_{2}$, and $\left[\!\left[10,2\!:\!1,4\right]\!\right]_{2}$ hybrid stabilizer codes.

Notably missing in our linear programming bounds is part of the nested code condition found in the linear programming bounds for hybrid stabilizer codes, namely that $A_{d}\leq A_{d}^{\left(a,a\right)}$ for all $d$. In fact we can construct a nonadditive hybrid code that violates this condition, as shown in the example below.

\begin{example}
We return to our $\left(\!\left(6,2\!:\!2,1\right)\!\right)_{2}$ nonadditive hybrid code from Example \ref{nonaddex}. The weight distributions for $\mathcal{C}_{a}$, $\mathcal{C}_{b}$, and $\mathcal{C}$ are
\begin{align*}
A^{\left(a,a\right)} & = \left[1,1,0,0,15,15,0\right] \\
A^{\left(b,b\right)} & = \left[1,0,1,0,11,16,3\right] \\
A & = \left[1,\frac{1}{4},\frac{1}{4},0,6,\frac{31}{4},\frac{3}{4}\right],
\end{align*}
where the weight distributions are the coefficients of the weight enumerators. These weight distributions clearly violate the inequality $A_{d}\leq A_{d}^{\left(a,a\right)}$.
\end{example}

Interestingly, we were unable to find any separation between our bounds with and without the nested code condition, suggesting the possibility that any hybrid code that meets these bounds must also satisfy this additional constraint. Since this condition is satisfied by any hybrid code constructed using the CWS framework, it seems that this comparable to the situation with quantum codes, where all known nonadditive codes meeting the linear programming bounds are CWS codes. Our bounds suggest the possibility of several nonadditive hybrid codes, such as $\left(\!\left(10,8\!:\!6,3\right)\!\right)_{2}$ and $\left(\!\left(13,8\!:\!3,3\right)\!\right)_{2}$ codes.

\section{Family of Single Error Detecting Codes}

In \cite{Rains1997}, Rains et al. constructed a $\left(\!\left(5,6,2\right)\!\right)_{2}$ nonadditive quantum code which was later extended to several families of $\left(\!\left(n,q^{n-3}<K<q^{n-2},2\right)\!\right)_{q}$ nonadditive codes with $n$ odd, see \cite{Aggarwal2008, Arvind2002, Feng2008, Rains1999a, Rigby2019, Smolin2007}. Rains \cite{Rains1999a} also showed that for any $\left(\!\left(n,K,2\right)\!\right)_{2}$ quantum code with odd $n$, \begin{equation*}K\leq2^{n-2}\left(1-\frac{1}{n-1}\right).\end{equation*} In particular, this disallows the existence of odd-lengthed $\left(\!\left(n,2^{n-2},2\right)\!\right)_{2}$ quantum codes.

Here we give a construction for a family of single error-detecting hybrid stabilizer codes such that $n$ is odd and $KM=2^{n-2}$, so these codes have the remarkable feature in that they allow one to squeeze in an additional classical bit. The generators of these codes are similar to the generators of the family of even-length stabilizer codes with parameters $\left[\!\left[n,n-2,2\right]\!\right]_{q}$, see \cite{Gottesman1997, Rains1999a}.

\begin{theorem}
For $n$ odd, there exists an $\left[\!\left[n,n-3\!:\!1,2\right]\!\right]_{2}$ genuine hybrid code with generators
\begin{equation*}
\left(\mkern-5mu
\begin{tikzpicture}[baseline=-.5ex]
\matrix[
  matrix of math nodes,
  column sep=.25ex, row sep=-.25ex
] (m)
{
X^{\otimes n-1} & X \\
Z^{\otimes n-1} & I \\
I^{\otimes n-1} & X \\
};
\draw[line width=1pt, line cap=round, dash pattern=on 0pt off 2\pgflinewidth]
  ([yshift=.2ex] m-2-1.south west) -- ([yshift=.2ex] m-2-2.south east);
\end{tikzpicture}\mkern-5mu
\right)
\end{equation*}
\end{theorem}
\begin{proof}

Recall that a number is said to have \emph{even parity} if it has an even number of 1's in its binary expansion. Let $J\subseteq\mathbb{F}_{2}^{n-1}$ be the set of even integers with even parity. We define two codes $\mathcal{C}_{0}$ and $\mathcal{C}_{1}$ as follows: $$\mathcal{C}_{0}=\left\{\frac{1}{2}\left(\ket{x}+\ket{\overline{x}}\right)\left(\ket{0}+\ket{1}\right)\middle| x\in J\right\},$$ $$\mathcal{C}_{1}=\left\{\frac{1}{2}\left(\ket{x}-\ket{\overline{x}}\right)\left(\ket{0}-\ket{1}\right)\middle| x\in J\right\}.$$ It is clear that the stabilizer of $\mathcal{C}_{0}$ is $\left\langle X^{\otimes n}, Z^{\otimes n-1}I, I^{\otimes n-1}X\right\rangle$ and that the stabilizer of $\mathcal{C}_{1}$ is $\left\langle X^{\otimes n}, Z^{\otimes n-1}I, -I^{\otimes n-1}X\right\rangle$. To show that our hybrid code has minimum distance 2, we note first that both $\mathcal{C}_{0}$ and $\mathcal{C}_{1}$ have minimum distance 2 when viewed as separate quantum codes. Thus we only need to look at how single-qubit Pauli errors affect the classical information. Consider two codewords $\ket{c_{i}^{\left(0\right)}}$ and $\ket{c_{j}^{\left(1\right)}}$, one from each quantum code. If $i\neq j$, it is clear that $\bra{c_{i}^{\left(0\right)}}E\ket{c_{j}^{\left(1\right)}}=0$ for any single-qubit Pauli error, since they will be linear combinations of disjoint sets of orthonormal basis vectors. Therefore, we can consider only the case when $i=j$.

Suppose that a single-qubit error occurs on the first $n-1$ qubits, that is $E=I_{2}^{\otimes\ell}\otimes E'\otimes I_{2}^{\otimes n-\ell-2}\otimes I_{2}$, for $\ell\in\left[n-1\right]$. Then since each of our codewords is separable between the first $n-1$ qubits and the last qubit, we can write
\begin{align*}
\bra{c_{i}^{\left(0\right)}}E\ket{c_{i}^{\left(1\right)}} & = \frac{1}{4}\left(\left(\bra{x}+\bra{\overline{x}}\right)\left(\bra{0}+\bra{1}\right)\right)E\left(\left(\ket{x}-\ket{\overline{x}}\right)\left(\ket{0}-\ket{1}\right)\right) \\
& = \frac{1}{4}\left(\left(\bra{x}+\bra{\overline{x}}\right)E'\left(\ket{x}-\ket{\overline{x}}\right)\right)\cdot\left(\left(\bra{0}+\bra{1}\right)\left(\ket{0}-\ket{1}\right)\right) \\
& = \frac{1}{4}\left(\left(\bra{x}+\bra{\overline{x}}\right)E'\left(\ket{x}-\ket{\overline{x}}\right)\right)\cdot 0 \\
& = 0.
\end{align*}
Similarly, if a single-qubit error occurs on the last qubit, that is $E=I_{2}^{\otimes n-1}\otimes E'$, we have
\begin{align*}
\bra{c_{i}^{\left(0\right)}}E\ket{c_{i}^{\left(1\right)}} & = \frac{1}{4}\left(\left(\bra{x}+\bra{\overline{x}}\right)\left(\bra{0}+\bra{1}\right)\right)E\left(\left(\ket{x}-\ket{\overline{x}}\right)\left(\ket{0}-\ket{1}\right)\right) \\
& = \frac{1}{4}\left(\left(\bra{x}+\bra{\overline{x}}\right)\left(\ket{x}-\ket{\overline{x}}\right)\right)\cdot\left(\left(\bra{0}+\bra{1}\right)E'\left(\ket{0}-\ket{1}\right)\right) \\
& = 0\cdot\left(\left(\bra{0}+\bra{1}\right)E'\left(\ket{0}-\ket{1}\right)\right) \\
& = 0.
\end{align*}
Thus the hybrid code given by $\mathcal{C}_{0}\oplus\mathcal{C}_{1}$ has minimum distance 2.

By a result of Rains \cite[Theorem 2]{Rains1999a}, for a general $\left(\!\left(n,K,2\right)\!\right)_{2}$ quantum code with $n$ odd, we have $$K\leq2^{n-2}\left(1-\frac{1}{n-1}\right).$$ In particular, this precludes the possibility of an $\left(\!\left(n,2^{n-2},2\right)\!\right)_{2}$ code for $n$ odd. Similarly, suppose that we could construct a code in our family using an $\left[\!\left[n_{q},k,d\right]\!\right]_{2}$ quantum code and an $\left[n_{c},m,d\right]_{q}$ classical code. Then we would have $n_{q}+n_{c}=n$, $k=n-3$, and $m=1$, and in particular, we have an $\left[\!\left[n_{q},n_{q}+n_{c}-3,2\right]\!\right]_{2}$ quantum code. By the quantum Singleton bound, we must have $n_{c}\leq 1$, forcing us to have a $\left[1,1,2\right]_{2}$ classical code, which of course does not exist. It follows that all of the codes in our family must be genuine hybrid codes.
\end{proof}

An interesting question is whether or not this family of hybrid codes are optimal, by which we mean do there exist odd-length $\left(\!\left(n,2^{n-3}\!:\!M,2\right)\!\right)_{2}$ codes with $M>2$? For small lengths ($n\leq19$) this family achieves the linear programming bounds for general hybrid codes given in Section \ref{lpb}, and we suspect that $M=2$ is optimal for all odd $n$.

\section{Families of Hybrid Codes from\\ Stabilizer Pasting}

In this section, we construct two families of single-error correcting hybrid codes that can encode one or two classical bits. An infinite family of nonadditive quantum codes was constructed by Yu et al. \cite{Yu2015} by pasting together (see \cite{Gottesman1996a}) the stabilizers of Gottesman's $\left[\!\!\!\:\left[2^{j},2^{j}-j-2,3\right]\!\!\!\:\right]_{\!\!\:2}$ codes \cite{Gottesman1996b} with the non-Pauli observables of the $\left(\!\left(9,12,3\right)\!\right)_{2}$ and $\left(\!\left(10,24,3\right)\!\right)_{2}$ nonadditive CWS codes \cite{Yu2007, Yu2008} which function in the same role as the Pauli stabilizers in stabilizer codes.

Below we give the generators of the hybrid codes originally given by Grassl et al. \cite{Grassl2017} that we will use in the construction of our families. The generators for the $\left[\!\left[7,1\!:\!1,3\right]\!\right]_{2}$ code was previously given in (\ref{gen7}), while those for the $\left[\!\left[9,2\!:\!2,3\right]\!\right]_{2}$, $\left[\!\left[10,3\!:\!2,3\right]\!\right]_{2}$, and $\left[\!\left[11,4\!:\!2,3\right]\!\right]_{2}$ hybrid stabilizer codes are (\ref{gen9}), (\ref{gen10}), and (\ref{gen11}) respectively:

\begin{equation}
\label{gen9}
\left(\mkern-5mu
\begin{tikzpicture}[baseline=-.65ex]
\matrix[
  matrix of math nodes,
  column sep=.25ex, row sep=-.25ex
] (m)
{
X & I & I & Z & Y & Z & X & X & Y \\
Z & X & I & Z & Y & X & Y & I & Z \\
I & Z & X & Z & Z & I & X & I & X \\
I & Z & Z & I & Y & X & X & Y & I \\
Z & Z & I & X & X & I & X & Z & I \\
Z & I & I & I & I & X & I & I & I \\
I & Z & I & I & I & I & X & I & I \\ 
};
\draw[line width=1pt, line cap=round, dash pattern=on 0pt off 2\pgflinewidth]
  ([yshift=.2ex] m-5-1.south west) -- ([yshift=.2ex] m-5-9.south east);
\end{tikzpicture}\mkern-5mu
\right)
\end{equation}

\begin{equation}
\label{gen10}
\left(\mkern-5mu
\begin{tikzpicture}[baseline=-.65ex]
\matrix[
  matrix of math nodes,
  column sep=.25ex, row sep=-.25ex
] (m)
{
X & X & I & Z & I & Z & Y & Z & Y & Z \\
X & I & Y & X & I & X & Z & X & X & Y \\
X & Z & X & Y & Z & Y & Y & I & I & Y \\
I & I & Z & Z & X & X & Y & Y & I & I \\
Z & I & I & I & Z & Z & X & X & I & X \\
Z & I & I & I & I & I & I & I & I & X \\
I & I & Z & Z & I & I & I & I & I & I \\ 
};
\draw[line width=1pt, line cap=round, dash pattern=on 0pt off 2\pgflinewidth]
  ([yshift=.2ex] m-5-1.south west) -- ([yshift=.2ex] m-5-10.south east);
\end{tikzpicture}\mkern-5mu
\right)
\end{equation}

\begin{equation}
\label{gen11}
\left(\mkern-5mu
\begin{tikzpicture}[baseline=-.65ex]
\matrix[
  matrix of math nodes,
  column sep=.25ex, row sep=-.25ex
] (m)
{
I & Z & X & I & X & Z & I & Z & X & X & X \\
I & Z & Z & X & I & I & Z & X & X & Y & Y \\
Z & I & I & Z & X & X & Z & X & X & X & I \\
X & X & I & X & Y & X & I & Y & Y & Y & X \\
Y & Y & I & X & X & Y & Y & Z & Y & I & Y \\
Z & I & I & I & I & I & I & I & X & I & I \\
I & Z & I & I & I & I & I & I & X & I & I \\ 
};
\draw[line width=1pt, line cap=round, dash pattern=on 0pt off 2\pgflinewidth]
  ([yshift=.2ex] m-5-1.south west) -- ([yshift=.2ex] m-5-11.south east);
\end{tikzpicture}\mkern-5mu
\right)
\end{equation}

Note that in each case, the generators above the dotted line define a pure $\left[\!\left[n,n-5,2\right]\!\right]_{2}$ quantum code.

The next theorem describes families of hybrid quantum codes. Notice that $2^{2m+5} \equiv 2^5 \pmod{3}$, so the length $n$ given in the theorem is well-defined.

\begin{theorem}  
Let $m$ be a nonnegative integer and $n$ a positive integer given by
$$n=\frac{2^{2m+5}-32}{3}+a,$$
where the parameter $a$ is a small positive integer that is specified below. Then there exists 
\begin{compactenum}[(a)]
\item an $\left[\!\left[n,n-2m-6\!:\!1,3\right]\!\right]_{2}$ hybrid code for $a=7$ and 
\item an $\left[\!\left[n,n-2m-7\!:\!2,3\right]\!\right]_{2}$ hybrid code for $a=9,10,11$.
\end{compactenum}
\end{theorem}
\begin{proof}
Roughly speaking, we construct our code by partitioning the first $\left(2^{2m+5}-32\right)\!/3$ qubits into disjoints sets, forming a perfect code on each partition, and use one of the four small hybrid codes on the remaining last $a$ qubits. These codes are then ``glued" to one another by using stabilizer pasting. Other than a small number of degenerate errors introduced by the small hybrid code that must be handled individually, each single-qubit Pauli error has a unique syndrome, allowing for the correction of any single-qubit error.

We will now describe the code construction in more detail. We 
take the $n=\left(2^{2m+5}-32\right)\!/3+a$ qubits and partition them into disjoint sets 
$$U_{m}\cup U_{m-1}\cup\cdots\cup U_{1}\cup V_{a},$$ 
where $\left\lvert U_{k}\right\rvert=2^{2k+3}$ and $\left\lvert V_{a}\right\rvert=a$.
The set $U_{m}$ contains the first $2^{2m+3}$ qubits, $U_{m-1}$ the next $2^{2m+1}$ qubits, and so forth. The final $a$ qubits are contained in $V_a$. 

Let $k$ be an integer in the range $1\le k\le m$. On the qubits in the set $U_{k}$, we can construct a stabilizer code of length $2^{2k+3}$ with $2k+5$ stabilizer generators, following Gottesmann~\cite{Gottesman1996b}. The $2k+5$ stabilizer generators are given as follows. Two of these generators are the tensor product of only Pauli-$X$ and $Z$ operators, which we call $X_{U_{k}}$ and $Z_{U_{k}}$ respectively. We define the other $2k+3$ stabilizers by
\begin{equation*}
\mathcal{S}_{j}^{k}=X^{h_{j}}Z^{h_{j-1}+h_{1}+h_{2k+3}},
\end{equation*}
for $j\in\left[2k+3\right]$. Here we let $h_{j}$ be the $j$-th row of the $\left(2k+3\right)\times2^{2k+3}$ matrix $H_{k}$, whose $i$-th column is the binary representation of $i$, $h_{0}$ is defined to be the all-zero vector, and $X^{h_{j}}=X^{h_{j,0}}X^{h_{j,1}}\dots X^{h_{j,2^{2k+3}-1}}$, with $Z^{h_{j}}$ defined similarly.

For the set $V_{a}$, let $H_{j}^{\mathcal{Q}}$ be the generators of the quantum stabilizer $\mathcal{S}_{\mathcal{Q}}$ of the length $a$ hybrid code defined by the generators in (\ref{gen7}), (\ref{gen9}), (\ref{gen10}), or (\ref{gen11}), and $H_{j}^{\mathcal{C}}$ be the generators of the classical stabilizer $\mathcal{S}_{\mathcal{C}}$ (since the length 7 hybrid code only has one generator in $\mathcal{S}_{\mathcal{C}}$, we can remove $H_{2}^{\mathcal{C}}$). The stabilizer can be pasted together as shown in (\ref{stabpastgen}), where suitable identity operators should be inserted in the blank spaces:

\begin{equation}
\label{stabpastgen}
\left(\mkern-5mu
\begin{tikzpicture}[baseline=-.65ex]
\matrix[
  matrix of math nodes,
  column sep=.25ex, row sep=-.25ex
] (m)
{
X_{U_{m}} & & & & & \\
Z_{U_{m}} & & & & & \\
S_{1}^{m} & X_{U_{m-1}} & & & & \\
S_{2}^{m} & Z_{U_{m-1}} & & & & \\
\vdots & \vdots & \ddots & & & \\
S_{2m-6}^{m} & S_{2m-8}^{m-1} & \cdots & & & \\
S_{2m-5}^{m} & S_{2m-7}^{m-1} & \cdots & X_{U_{2}} & & \\
S_{2m-4}^{m} & S_{2m-6}^{m-1} & \cdots & Z_{U_{2}} & & \\
S_{2m-3}^{m} & S_{2m-5}^{m-1} & \cdots & S_{1}^{2} & X_{U_{1}} & \\
S_{2m-2}^{m} & S_{2m-4}^{m-1} & \cdots & S_{2}^{2} & Z_{U_{1}} & \\
S_{2m-1}^{m} & S_{2m-3}^{m-1} & \cdots & S_{3}^{2} & S_{1}^{1} & H_{1}^{\mathcal{Q}} \\
S_{2m}^{m} & S_{2m-2}^{m-1} & \cdots & S_{4}^{2} & S_{2}^{1} & H_{2}^{\mathcal{Q}} \\
S_{2m+1}^{m} & S_{2m-1}^{m-1} & \cdots & S_{5}^{2} & S_{3}^{1} & H_{3}^{\mathcal{Q}} \\
S_{2m+2}^{m} & S_{2m}^{m-1} & \cdots & S_{6}^{2} & S_{4}^{1} & H_{4}^{\mathcal{Q}} \\
S_{2m+3}^{m} & S_{2m+1}^{m-1} & \cdots & S_{7}^{2} & S_{5}^{1} & H_{5}^{\mathcal{Q}} \\
 & & & & & H_{1}^{\mathcal{C}} \\
 & & & & & H_{2}^{\mathcal{C}} \\
};
\draw[line width=1pt, line cap=round, dash pattern=on 0pt off 2\pgflinewidth]
  ([yshift=.2ex] m-15-1.south west) -- ([yshift=.2ex] m-15-6.south east);
\end{tikzpicture}\mkern-5mu
\right)
\end{equation}

Suppose that we have an single-qubit Pauli error on the block $U_{m}$. Since the code is pure, the syndrome of each error will be distinct and such that the Pauli-$X$, $Y$, and $Z$ sydromes will start with $01$, $11$, and $10$ respectively. However, this leaves all of the syndromes starting with $00$ unused, so Pauli-$X$, $Y$, and $Z$ errors on the block $U_{m-1}$ will have distinct syndromes starting with $0001$, $0011$, and $0010$ respectively. Continuing on, any single-qubit Pauli error occurring on the block $U_{k}$ will have a distinct syndrome starting with $2\left(m-k\right)$ $0$s.

All of the syndromes of errors occurring on the block $V_{a}$ start with $2m$ $0$s. Here our code is not pure, but it is almost pure, with the only degenerate errors being the weight 2 errors in $\mathcal{S}_{\mathcal{C}}$. For example, when $V_{a}$ has 11 qubits, it will have three weight 1 degenerate errors: $Z_{1}$ (a Pauli-$Z$ on the first qubit of the block), $Z_{2}$, and $X_{9}$, each with the syndrome $00011$ (preceeded by $2m$ zeros). If we measure this syndrome, we apply the operator $ZZIIIIIIXII$ to the state, which maps the original codeword to itself up to a global phase. Note, however, that while this global phase is the same for codewords of the same inner code for a given error, it may differ for codewords from different inner codes. In fact, this is exactly what prevents the outer code from being a distance 3 quantum code rather than a distance 3 hybrid code. The argument for when $V_{a}$ has 7, 9, and 10 qubits is similar.

Since we know how to correct any single-qubit Pauli error based on its syndrome, each of the codes must have minimum distance 3.
\end{proof}

Here we show that these hybrid codes are better than optimal quantum stabilizer codes using a result of Yu et al. \cite{Yu2013}.

\begin{proposition}
Let $m$ be a nonnegative integer and $n$ a positive integer given by
$$n=\frac{2^{2m+5}-32}{3}+a,$$
where $a\in\left\{7,9,10,11\right\}$. Then there does not exist an $\left[\!\left[n,n-2m-5,3\right]\!\right]_{2}$ stabilizer code.
\end{proposition}
\begin{proof}
When $a=7,9,10$, we have 
\begin{align*}
n & = \frac{2^{2m+5}-32}{3}+a \\
& = \frac{2^{2m+5}-8}{3}+\left(a-8\right) \\
& = \frac{8}{3}\left(4^{m+1}-1\right)+\left(a-8\right).
\end{align*}
By a result of Yu et al. \cite[Theorem 1]{Yu2013}, distance 3 stabilizer codes with lengths of the form $$\frac{8}{3}\left(4^{k}-1\right)+b,$$ where $b\in\left\{-1,1,2\right\}$, can exist if and only if $$2m+5\geq \left\lceil\log_{2}\!\left(3n+1\right)\right\rceil+1.$$ But in this case we have
\begin{align*}
\left\lceil\log_{2}\!\left(3n+1\right)\right\rceil+1 & = \left\lceil\log_{2}\!\left(2^{2m+5}+3a-31\right)\right\rceil+1\\
 & > \left\lceil\log_{2}\!\left(2^{2m+5}-2^{2m+4}\right)\right\rceil+1 \\
 & = 2m+5,
\end{align*}
so when $a=7,9,10$, there is no distance 3 stabilizer code of length $n$.

When $a=11$, a different case of \cite[Theorem 1]{Yu2013} applies, so distance 3 stabilizer codes with lengths of this form can exist if and only if $$2m+5\geq \left\lceil\log_{2}\!\left(3n+1\right)\right\rceil.$$ However, this gives us
\begin{align*}
\left\lceil\log_{2}\!\left(3n+1\right)\right\rceil & = \left\lceil\log_{2}\!\left(2^{2m+5}+2\right)\right\rceil\\
 & > \left\lceil\log_{2}\!\left(2^{2m+5}\right)\right\rceil \\
 & = 2m+5,
\end{align*}
so when $a=11$, there is likewise no distance 3 stabilizer code of length $n$.
\end{proof}

As with our family of error-detecting hybrid codes, it would be interesting to know whether any of these codes meet the linear programming bounds from Section \ref{lpb}. Since none of the hybrid codes we started with meet these bounds, it is doubtful that any of the hybrid codes constructed from stabilizer pasting would also meet this bound, leaving it unclear whether or not these codes are optimal among all hybrid codes.

\section{Conclusion and Discussion}
In this paper we have proven some general results about hybrid codes, showing that they can always detect more errors than comparable quantum codes. Furthermore we proved the necessity of impurity in the construction of genuine hybrid codes. Additionally, we generalized weight enumerators for hybrid stabilizer codes to nonadditive hybrid codes, allowing us to develop linear programming bounds for nonadditive hybrid codes. Finally, we have constructed several infinite families of hybrid stabilizer codes that provide an advantage over optimal stabilizer codes.

Both of our families of hybrid codes were inspired by the construction of nonadditive quantum codes. In hindsight this is not very surprising, as the examples of hybrid codes with small parameters given by Grassl et al. \cite{Grassl2017} were constructed using a CWS/union stabilizer construction. Most interesting is that all known good nonadditive codes with small parameters have a hybrid code with similar parameters. This would suggest that looking at larger nonadditive codes such as the quantum Goethals-Preparata code \cite{Grassl2008} or generalized concatenated quantum codes \cite{Grassl2009} might be helpful in constructing larger hybrid codes. Alternatively, it may be possible to use the existence of hybrid codes to point to where nonadditive codes may be found. For instance the existence of an $\left[\!\left[11,4\!:\!2,3\right]\!\right]_{2}$ hybrid code suggests a nonadditive code with similar parameters might exist.

As previously suggested by Grassl et al. \cite{Grassl2017}, one possible way to construct new hybrid codes with good parameters is to start with degenerate quantum codes with good parameters. Another possible approach to constructing new hybrid stabilizer codes is to find codes such that there are few small weight errors that are in the normalizer but not in the stabilizer, and then add those small weight errors to the generating set of the stabilizer to get a degenerate code. Here, the original code becomes the outer code of the hybrid code and the degenerate code the inner code.





\ifCLASSOPTIONcaptionsoff
  \newpage
\fi

\end{document}